\documentclass{llncs}
\usepackage{fullpage}
\usepackage{amsmath, amssymb}
\usepackage{verbatim}

\usepackage{hyperref}
\hypersetup{colorlinks=true,urlcolor=blue,linkcolor=blue,citecolor=blue}

\usepackage{tikz}
\usetikzlibrary{arrows}

\newcommand{\qa}{q_{A|\emptyset}}
\newcommand{\qb}{q_{B|\emptyset}}
\newcommand{\qab}{q_{A|B}}
\newcommand{\qba}{q_{B|A}}

\title{Complete Submodularity Characterization in the Comparative Independent Cascade Model}
\author{Wei Chen \inst{1} \and Hanrui Zhang \inst{2}}
\institute{Microsoft Research, Beijing, China \\ \email{weic@microsoft.com} \and Duke University, Durham, NC, USA \\ \email{hrzhang@cs.duke.edu}}

\begin{document}

\maketitle

\begin{abstract}
    We study the propagation of comparative ideas or items in social networks. A full characterization for submodularity in the comparative independent cascade (Com-IC) model of two-idea cascade is given, for competing ideas and complementary ideas respectively, with or without reconsideration. We further introduce One-Shot model where agents show less patience toward ideas, and show that in One-Shot model, only the strongest idea spreads with submodularity.
\end{abstract}

\keywords{influence submodularity, comparative independent cascade, influence maximization}

\section{Introduction}

    Propagation of information in social networks has been extensively studied over the past decades, along with its most prominent algorithmic aspect - influence maximization. The cascade procedure of ideas in a network is usually modeled by a stochastic process, and influence maximization seeks to maximize the expected influence of a certain idea by choosing $k$ agents (the seed set) in the network to be early adopters of the idea. The seed set then initiates the propagation through the network structure.
    
    Influence maximization is proven to be $\mathsf{NP}$-hard \cite{kempe2003maximizing} in almost any non-trivial setting. Most research therefore focuses on approximation algorithms, some particularly successful ones out of which are based on the celebrated $(1 - \frac1e)$-approximate submodular maximization \cite{nemhauser1978analysis}. Submodularity of influence in the seed set therefore plays a central role in such optimization.\footnote{We say a function $f: 2^U \rightarrow \mathbb{R}$ is submodular, if for any $S \subseteq U$, $a, b \in U$, $f(S) + f(S \cup \{a, b\}) \le f(S \cup \{a\}) + f(S \cup \{b\})$.}
    
    Nevertheless, submodularity appears harder to tract when there are multiple ideas interacting with each other. Most prior work focuses on single-idea cascade, or completely competing propagation of ideas. These models somewhat fail in modeling real world behavior of agents. Lu et al.\ \cite{lu2015competition} introduce a general model called {\em comparative independent cascade (Com-IC) model}, 
    which covers the entire spectrum of two item cascades from full competition to full complementarity. 
    This full spectrum is crucially characterized by four probability parameters called {\em global adoption probabilities (GAP)}, and their
	    space is called the GAP space.
    However, they only provide submodularity analysis in a few marginal cases of the entire GAP space, and a full submodularity 
    characterization for the entire GAP space is left
    as an open problem discussed in their conclusion section.
    
    \paragraph{Our contribution.} In this paper, we provide a full characterization of the submodularity of the Com-IC model in both the
    mutually competing case and the mutually complementary case, with or without reconsideration (Theorems~\ref{thm:competing}, \ref{thm:complementary_self_without_reconsideration}, \ref{thm:complementary_self_with_reconsideration} and \ref{thm:complementary_cross}). Our results show that in the entire continuous GAP space, the parameters satisfying submodularity
	    only has measure zero.
    Next, we introduce a slightly modified One-Shot model for the mutual competing case where agents are less patient: they would reject 
	    all items if they get influenced by but fail to adopt any item.
    We provide the full submodularity
    characterization of the parameter space for this model (Theorem~\ref{thm:one-shot}), which contains a nontrivial half space 
    satisfying submodularity, contrasting the result for the Com-IC model.
    Our techniques for establishing these characterization 
	    results may draw separate interests from the technical aspect for the study of submodularity for
	    various influence propagation models.
    
    \paragraph{Related work.} Single-idea models, where there is only one propagating entity for social network users to adopt, has been thoroughly studied. Some examples are the classic Independent Cascade (IC) and Linear Thresholds (LT) models \cite{kempe2003maximizing}. Some other work studies pure competition between ideas. See, e.g.\ \cite{bharathi2007competitive,borodin2010threshold,budak2011limiting,chen2011influence,he2012influence,lu2013bang}. Beside competing settings, Datta et al.\ \cite{datta2010viral} study influence maximization of independently propagating ideas, and Narayanam et al.\ \cite{narayanam2012viral} discuss a perfectly complementary setting, which is extended in \cite{lu2015competition}.
    
\section{The Model}
\label{sec:model}

    We first recapitulate the independent cascade model for comparative ideas (Com-IC).

    First recall that in the classic Independent Cascade (IC) model, the social network is described by a directed graph $G = (V, E, p)$ with probabilities $p: E \rightarrow [0, 1]$ on each edge. Each vertex in $V$ stands for an agent, an edge for a connection, whose strength is characterized by the associated probability. Cascading proceeds at each time step $0, 1, \dots$. At time $0$, only the seed set is active. At time $t$, each vertex $u$ activated at time $t - 1$ tries to activate its neighbor $v$, and succeeds with probability $p(u, v)$. The procedure ends when no new vertices are activated at some time step.

    \paragraph{Basic states and transition.} In comparative IC (Com-IC henceforth) model, there are two ideas, $A$ and $B$, spreading simultaneously in the network, and therefore 9 basic states of each vertex:
    \[
        \{A\textrm{-idle}, A\textrm{-adopted}, A\textrm{-rejected}\} \times \{B\textrm{-idle}, B\textrm{-adopted}, B\textrm{-rejected}\}.
    \]

    Items propagate along the edges in the same way. That is, when some vertex $u$ is activated by $A$, it proposes $A$ to all its neighbors, and the proposal reaches its neighbor $v$ with probability $p(u, v)$. Additionally, when an $A$-proposal reaches an $A$-idle vertex $u$, if $u$ is previously $B$-adopted, it adopts $A$ w.p.\ $\qab$. Otherwise, it adopts $A$ w.p.\ $\qa$. The rules for idea $B$ is totally symmetric. The four probabilities, $\qa, \qb, \qab, \qba$, therefore fully characterize strengths of the two ideas and the relationship between them: when $A$ and $B$ are mutually competing ideas, $\qa \ge \qab$ and $\qb \ge \qba$; when they are mutually complementary ideas, $\qa \le \qab$ and $\qb \le \qba$.

    \paragraph{Reconsideration.} For two complementary items correlated in certain ways, adoptation of one item may result in reconsideration of the other which has been rejected before. This phenomenon is modeled by adding a suspended state and a reconsideration process. For two complementary items $A$ and $B$, suppose $A$ reaches a vertex $u$ first. If $u$ adopts $A$, then everything works in the same way, i.e., the state of $u$ becomes $A\textrm{-adopted} \times B\textrm{-idle}$. It then adopts $B$ w.p.\ $\qba$ and rejects w.p.\ $1 - \qba$. But when $u$ rejects $A$, instead of becoming $A\textrm{-rejected}$, it enters a state called $A\textrm{-suspended}$. When $B$ reaches $u$ later, $u$ adopts $B$ w.p.\ $\qb$ and rejects w.p.\ $1 - \qb$. Moreover, if $u$ adopts $B$, it reconsiders $A$ and adopts w.p.\ $\rho_A$. Only after reconsideration, $u$ becomes $A\textrm{-adopted}$ or $A\textrm{-rejected}$. The rules for $B$ are again symmetric. In Com-IC model, it is further required that the parameters satisfy certain conditions such that at any vertex, it does not matter which item makes its proposal first. Namely, for $\rho_A$,
    \[
        \qa + (1 - \qa) \qb \rho_A = (1 - \qb) \qa + \qb \qab.
    \]
    In the above condition, both sides can be expressed as a probability that a vertex adopts $A$. In the left hand side, item $A$ makes a proposal first, and then $B$ does. The probability of adopting $A$ is therefore the sum of the probability of an instant adoptation upon $A$'s proposal, and the probability of adopting $B$ and a successful reconsideration following. In the right hand side, $B$ makes a proposal first and then $A$ does. The probability of adopting $A$ is then the sum of the probability that $B$ fails and $A$ succeeds, and the probability that both succeed. A similar rule exists for $\rho_B$. As a result, $\rho_A$ and $\rho_B$ are determined by $\qa$, $\qb$, $\qab$ and $\qba$. As we will see, this independence of order greatly simplifies the analysis of the propagation procedure.

    These four probability parameters ($\qa$, $\qb$, $\qab$, $\qba$) are referred to as {global adoption probabilities (GAP)}, and their space as the GAP space.
    
    For tie-breaking, we generate a random ordering of all in-going edges for each vertex, and let proposals which reach at the same time try according to that order. If a vertex adopts two ideas at a same time step, it proposes the two ideas to its neighbors in the order adopted. We refer interested readers to \cite{lu2015competition} for more details of Com-IC model.

    \paragraph{On power of edge probabilities.} Although probabilities on edges seem to make the model more complicated, we note that essentially they do not affect the submodularity of the model. In fact, to show that a group of GAP guarantees submodularity for any network, one may partially realize all randomness on edges, argue submodularity in each realized world, and show submodularity in the original network by taking expectation. As a result, submodularity with edge probabilities is exactly equivalent to that without edge probabilities. In the rest of the paper, we always consider probabilities on edges partially realized, and therefore assume the probability of any edge is $1$.

\section{Notations}

    Let the set of possible worlds (the complete state of the network and vertices after fixing all randomness) be $\mathcal{W}$. For a possible world $W \in \mathcal{W}$, $A$-seed set $S_A$ and $B$-seed set $S_B$ (unless otherwise specified), let $\sigma_A(S_A, S_B, W)$ (resp.\ $\sigma_B(S_A, S_B, W)$) be the number of vertices which adopt $A$ (resp.\ $B$) at the end of cascading in possible world $W$. $\sigma_A(S_A, S_B) = \mathbb{E}[\sigma_A(S_A, S_B, W)]$ (resp.\ $\sigma_B(S_A, S_B) = \mathbb{E}[\sigma_B(S_A, S_B, W)]$) then stands for the expected influence of $A$ (resp.\ $B$) after cascading. Similarly, let $\sigma_A^u(S_A, S_B, W)$ be $1$ if $A$ affects $u$ in $W$, and $0$ if not, and $\sigma_A^u(S_A, S_B) = \mathbb{E}[\sigma_A^u(S_A, S_B, W)]$ the probability that $A$ affects $u$. Parameters are ignored when in clear context.

\section{Submodularity in the Mutually Competing Case}

    \begin{figure}[t]
    \centering
    \begin{tikzpicture}
        \tikzset{vertex/.style = {shape=circle,draw,minimum size=1.5em}}
        \tikzset{edge/.style = {->,> = latex'}}
        \node[vertex] (u) at (0, 0) {$u$};
        \node[vertex] (a) at (2, 0) {$a$};
        \node[vertex] (b) at (4, 0) {$b$};
        \node[vertex] (w) at (2, -2) {$w$};
        \node[vertex] (v) at (4, -1) {$v$};
        \node[vertex] (d1) at (0, -1) {};
        \node[vertex] (d2) at (0, -2) {};
        \node[vertex] (d3) at (0, -3) {};
        \node[vertex] (t) at (2, -4) {$t$};
        \draw[edge] (u) to (d1);
        \draw[edge] (d1) to (d2);
        \draw[edge] (d2) to (d3);
        \draw[edge] (d3) to (t);
        \draw[edge] (a) to (w);
        \draw[edge] (w) to (t);
        \draw[edge] (b) to (v);
        \draw[edge] (v) to (w);
    \end{tikzpicture}
    \caption{Counterexample used in the proofs of Theorem~\ref{thm:competing} and Theorem~\ref{thm:complementary_self_without_reconsideration}.}
    \label{fig:competing_and_self}
    \end{figure}
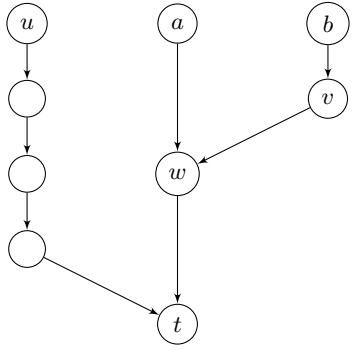

    Recall that when the two ideas are competing, we have $\qa \ge \qab$, $\qb \ge \qba$. We are naturally interested in submodularity of $\sigma_A(S_A, S_B)$ in $S_A$ fixing $S_B$. It turns out that this kind of submodularity is guaranteed only in a 0-measure subset of the parameter space. Formally, we have the following theorem:
    \begin{theorem}[Submodularity Characterization for the Mutually Competing Case]
    \label{thm:competing}
        When the two ideas are mutually competing, for a fixed $S_B$, $\sigma_A$ is submodular in $S_A$ whenever one of the following holds:
        \begin{itemize}
            \item $\qa = 1$,
            \item $\qa = \qab$,
            \item $\qb = \qba$.
        \end{itemize}
        And when none of these conditions hold, submodularity is violated, i.e., there exists $(G, S_A, S_B, u, v)$ such that for each group of $(\qa, \qb, \qab, \qba)$ not satisfying the above conditions,
        \[
            \sigma_A(S_A, S_B) + \sigma_A(S_A \cup \{u, v\}, S_B) > \sigma_A(S_A \cup \{u\}, S_B) + \sigma_A(S_A \cup \{v\}, S_B).
        \]
    \end{theorem}

    \begin{proof}
        First we prove the negative (non-submodular) half of the theorem by given an counterexample, illustrated in Figure~\ref{fig:competing_and_self}. The basic seed sets for $A$ and $B$ are $S_A = \{a\}$ and $S_B = \{b\}$ respectively. In order to show non-submodularity, we consider the marginals of $u$ at $t$ when $v$ is an $A$-seed and when $v$ is not. 
        
        Note that considering submodularity at a single vertex suffices for establishing a global proof, since we could duplicate the vertex such that it dominates the expected influence. Also, we assume $p(u, v) = 1$ for each $(u, v) \in E$, since, as discussed above, all positive (submodularity) proofs can be partially derandomized and done in each partial possible world, and for counterexamples, we simply set the probabilities to be $1$.
        
        Formally, define
        \begin{align*}
            M_1 & = \sigma_A^t(S_A \cup\{u\}, S_B) - \sigma_A^t(S_A, S_B), \\
            M_2 & = \sigma_A^t(S_A \cup \{u, v\}, S_B) - \sigma_A^t(S_A \cup \{v\}, S_B).
        \end{align*}
        Submodularity is violated if we show $M_1 < M_2$. We now calculate $M_1$ and $M_2$ separately. When $v$ is not a seed, $u$ has a marginal at $t$ iff $a$ fails to activate $w$ and idea $A$ succeeds in affecting $t$ from $u$. This is because if $w$ is activated and fails to activate $t$, then $t$ will become $A$-rejected and never accept any $A$-proposal. So $M_1$ is exactly the probability that (1) $a$ does not activate $w$ (with probability $1 - \qa$), (2) $b$ does not activate $t$ and $u$ activates $t$ (with probability $(1 - \qb^3) \qa^4$), or $b$ activates $t$ and $u$ also activates $t$ (with probability $\qb^3 \qa^3 \qab$). That is,
        \[
            M_1 = (1 - \qa)[(1 - \qb^3)\qa^4 + \qb^3 \qa^3 \qab].
        \]
        Similarly, when $v$ is an $A$-seed, $u$ has a marginal if (1) $a$ does not activate $w$ (with probability $1 - \qa$), (2) $b$ does not activate $t$ and $u$ activates $t$ (with probability $(1 - \qb \qba \qb)\qa^4$), or $b$ activates $t$ and $u$ also activates $t$ (with probability $\qb \qba \qb \qa^3 \qab$). We have
        \[
            M_2 = (1 - \qa)[(1 - \qb \qba \qb)\qa^4 + \qb \qba \qb \qa^3 \qab].
        \]
        Taking the difference, we get
        \[
            M_2 - M_1 = \qa^3 \qb^2 (1 - \qa)(\qab - \qa)(\qba - \qb).
        \]
        It is easy to see, when none of the conditions listed in Theorem~\ref{thm:competing} hold, $M_2 - M_1 > 0$, and $\sigma_A$ is not submodular in the seed set of $A$.\footnote{Note that when $A$ and $B$ are competing, $\qab - \qa \ne 0 \Rightarrow \qa \ne 0$, and $\qba - \qb \ne 0 \Rightarrow \qb \ne 0$.}

        We now show case by case, that whenever one of the conditions holds, $\sigma_A$ is submodular in the seed set of $A$.
        \begin{itemize}
            \item
                $\qa = 1$. Consider an equivalent formulation of the model: each vertex $u$ draws two independent numbers uniformly at random from $[0, 1]$, denoted by $\alpha_A(u)$ and $\alpha_B(u)$ respectively. When an $A$-proposal reaches an $(A\textrm{-idle}, B\textrm{-idle})$ or $(A\textrm{-idle}, B\textrm{-rejected})$ vertex $u$, if $\alpha_A(u) \le \qa$, $u$ will accept $A$. When an $A$-proposal reaches an $(A\textrm{-idle}, B\textrm{-adopted})$ vertex $u$, if $\alpha_A(u) \le \qab$, $u$ will accept $A$. The rules for $B$ are symmetric.

                After fixing all randomness, each vertex has two attributes for ideas $A$ and $B$ respectively. That is, each vertex $u$ can be in exactly one state out of
                \begin{align*}
                    & \{\alpha_A(u) \le \qab, \qab < \alpha_A(u) \le \qa, \qa < \alpha_A(u)\} \times \\
                    & \{\alpha_B(u) \le \qba, \qba < \alpha_B(u) \le \qb, \qb < \alpha_B(u)\}.
                \end{align*}

                We show that in any possible world $W$, if $\sigma_A^t(S_A \cup \{u, v\}, S_B, W) = 1$, then $\sigma_A^t(S_A \cup \{u\}, S_B, W) + \sigma_A^t(S_A \cup \{v\}, S_B, W) \ge 1$. That is, if $t$ is reachable by $A$ when $u$ and $v$ are both $A$-seeds, then it is reachable by $A$ when $u$ or $v$ alone is an $A$-seed. Submodularity then follows from monotonicity of $\sigma_A^t(S_A, S_B, W)$ in $S_A$ and convex combination of possible worlds.
                
                Let $p = (w_1, \dots, w_k)$ be the $A$-path which reaches $t$ when $u$ and $v$ are both $A$-seeds, where $w_1$ is an $A$-seed, and $w_k = t$. W.l.o.g.\ $v \notin p$. We argue that for each $w \in p$, if $w$ is not $B$-adopted by the time $A$ arrives when $u$ and $v$ are both $A$-seeds, then $w$ is not $B$-adopted by the time $A$ arrives when only $u$ is an $A$-seed, and as a result, $p$ remains $A$-affected even if $v$ is not an $A$-seed. Suppose not. Let $w$ be the vertex closest to $w_1$ on $p$, which becomes affected by $B$ when $v$ is not a seed, $p'$ be the $B$-path through which $w$ is affected by $B$. Let $x$ be the closest vertex to the $B$-seed on $p'$, which is affected by $A$ at the time the $B$-proposal arrives when $v$ is an $A$-seed, and is affected by $B$ when $v$ is not a seed (such a vertex must exist). Then because $\qa = 1$, the subpath from $x$ to $t$, $[x, w] \subseteq p'$ and $[w, t] \subseteq p$, must be completely $A$-affected when $v$ is an $A$-seed, through which item $A$ reaches $t$ earlier than through $p$, a contradiction.

                Now since each vertex $w \in p$ which is not affected by $B$ when $v$ is an $A$-seed remains not affected when $v$ is not, idea $A$ can pass through the entire path $p$ from some seed vertex to $t$ just like when $v$ is an $A$-seed, so $t$ is still $A$-affected. In other words, w.l.o.g.\ $\sigma_A^t(S_A \cup \{u\}, S_B, W) = 1$.
            \item
                $\qa = \qab$. $B$ does not affect the propagation of $A$. Again the propagation of $A$ is equivalent as an IC procedure, and submodularity follows directly.
            \item
                $\qb = \qba$. We use the possible world model discussed in the first bullet point. Still, let $p = \{w_1, \dots, w_k\}$ be the path through which $t$ is affected by $A$ when both $u$ and $v$ are $A$-seeds, and w.l.o.g.\ $v \notin p$. We apply induction on $i$ to prove that $A$ reaches $w_i$ still at the $(i - 1)$-th time slot when $v$ is not an $A$-seed.

                When $i = 1$, the statement holds evidently as $w_1$ is an $A$-seed. Assume at time $i - 1$, $w_i$ has just been reached by $A$ and become $A$-adopted. Since the propagation of $B$ is not affected by the $A$ seed set or propagation, $w_{i + 1}$ is in the same state w.r.t.\ $B$ as when $v$ is also a seed, so the $A$-proposal to $w_{i + 1}$ from $w_i$ ends up just in the same way, and $w_{i + 1}$ becomes $A$-adopted at time $i$. As a result, $t$ is eventually $A$-adopted, i.e.\ $\sigma_A^t(S_A \cup \{u\}, S_B, W) = 1$.
        \end{itemize}
        \qed
    \end{proof}

\section{Submodularity in the Mutually Complementary Case}

    When the two ideas are complementary, i.e.\ when $\qa \le \qab$ and $\qb \le \qba$, enlarging the seed set of one idea helps the propagation of both the idea itself and that of the other idea. We discuss in this section the self and cross effect of the seed set of an idea, with or without reconsideration.

    \subsection{Self Submodularity}
    
        Fixing $S_B$, we are interested in submodularity of $\sigma_A$ in $S_A$, i.e., submodularity of the influence of some idea w.r.t.\ its own seed set, fixing the seed set of the other idea.

        \begin{theorem}[Self-Submodularity Characterization for the Mutually Complementary Case without Reconsideration]
        \label{thm:complementary_self_without_reconsideration}
            When the two ideas are complementary and no reconsideration is allowed, for a fixed $S_B$, $\sigma_A$ is submodular in $S_A$ whenever one of the following holds:
            \begin{itemize}
                \item $\qa = 0$,
                \item $\qb = 0$,
                \item $\qa = \qab$,
                \item $\qb = \qba$.
            \end{itemize}
            And when none of these conditions hold, submodularity is violated, i.e., there exists $(G, S_A, S_B, u, v)$ such that for each group of $(\qa, \qb, \qab, \qba)$ not satisfying the above conditions,
            \[
                \sigma_A(S_A, S_B) + \sigma_A(S_A \cup \{u, v\}, S_B) > \sigma_A(S_A \cup \{u\}, S_B) + \sigma_A(S_A \cup \{v\}, S_B).
            \]
        \end{theorem} 

        \begin{proof}
            We first show the negative part. Recall that in the proof of Theorem~\ref{thm:competing}, we calculate that for the graph in Figure~\ref{fig:competing_and_self},
            \[
                M_2 - M_1 = \qa^3 \qb^2 (1 - \qa)(\qab - \qa)(\qba - \qb),
            \]
            which remains exactly the same no matter whether $A$ and $B$ are competing or complementary. If none of the conditions in Theorem~\ref{thm:complementary_self_without_reconsideration} hold, then $M_2 - M_1 > 0$, and $\sigma_A^t$ is not submodular in the seed set of $A$.\footnote{Note that when $A$ and $B$ are complementary, $\qab - \qa \ne 0 \Rightarrow 1 - \qa \ne 0$.}

            Now we prove case by case the positive cases.
            \begin{itemize}
                \item
                    $\qa = 0$. The fact that $\qa = 0$ means that $A$ spreads only by following $B$.
                    We use the same notations as in the proof of Theorem~\ref{thm:competing}. Assume that in possible world $W$, when both $u$ and $v$ are $A$-seeds, $t$ is affected by $A$ (or $\sigma_A^t(S_A \cup \{u, v\}, S_B, W) = 1$), and let $p = \{w_1, \dots, w_k\}$ be the shortest path through which $A$ reaches $t$, where w.l.o.g.\ $v \notin p$. 
                    Note that here by shortest path we mean not only that the length
                    	of path $p$ is the shortest, but also that following the
                    	tie-breaking order of possible world $W$, this is the first
                    	path through which $A$ could reach $t$.

					Consider first that $S_A\cup \{u,v\}$ is the $A$-seed set.
					Since $p$ is the shortest path from any $A$ seed to $t$, 
						there is no other node on path $p$ that is an $A$ seed, 
						and $A$ has to pass through $p$ to reach $t$.
					Moreover, since 
						$A$ cannot propagate by itself and has to rely on the
						help of $B$ adoptions, we know that for all nodes from $w_2$
						on path $p$, $B$ has to arrive at these nodes before
						$A$ does in the possible world $W$, so that the adoptions
						of $B$ on the path help the propagation of $A$ along the path.
					This means that in the possible world $W$, for every node
						$w\in \{ w_2, \ldots, w_k\}$, $w$ adopts $B$ based on its
						$\qb$ condition, independent of $A$.
					Consider $w_2$ now, since $w_2$ is an out-neighbor of the
						$A$-seed $w_1$, then in order for $B$ to reach $w_2$ first,
						either $w_2$ itself is a $B$ seed, or $w_2$ is an
						out-neighbor of a $B$ seed and the tie-breaking order 
						in $W$ is such that $B$ arrives at $w_2$ first.
					We now consider that $S_A \cup \{u\}$ is the $A$-seed set.
					Since $v\not\in p$, we have $w_1 \in S_A \cup \{u\}$.
					By the above argument on $w_2$, we know that at $w_2$
						$B$ still arrives before $A$ does and $w_2$ adopts $B$.
					Then following the path $p$ from $w_2$, we know that all nodes
						on path $p$
						will adopt $B$ independent of $A$, since they
						all adopt $B$ based on their $\qb$ condition alone.
					Therefore, when $A$ arrives at $w_2$ from $w_1$, $w_2$ has already
						adopted $B$, which will help $w_2$ adopt $A$.
					Similarly, when $A$ arrives at $w_j$ ($j\ge 2$) along path $p$,
						$B$ has already arrived at $w_j$ and would help $w_j$ to
						adopt $A$.
					We remark that there is no other way that $A$ could arrive at
						$w_j$ through another path earlier than $B$, since otherwise
						that would either be instead the shortest path for $A$ to reach $t$, or stop $A$ from passing through $p$.
					Therefore, $A$ would still reach $t=w_k$, when
						$S_A \cup \{u\}$ is the $A$-seed set, i.e.
						$\sigma_A^t(S_A \cup \{u\}, S_B, W) = 1$.
                    This is enough to show the submodularity of $\sigma_A$
                    	with respect to $S_A$.
                \item 
                    $\qb = 0$. That is, $B$ spreads only through $A$-adopted vertices, and thus does not affect the propagation of $A$. The equivalent IC cascade procedure gives submodularity directly.
                \item
                    $\qa = \qab$. Again, $B$ does not affect $A$, and submodularity is trivial.
                \item
                    $\qb = \qba$. The proof is totally similar to the last bullet point in the proof of Theorem~\ref{thm:competing}.
            \end{itemize}
        \qed
        \end{proof}
        
        \begin{note}
            The counterexample used in the proof of Theorem~\ref{thm:complementary_self_without_reconsideration} is exactly the same as that used in the proof of Theorem~\ref{thm:competing}. This versatility of the counterexample comes from the factor $(\qa - \qab)(\qb - \qba)$. In each case, $\qa - \qab$ and $\qb - \qba$ are of the same sign.
        \end{note}

        \begin{theorem}[Self-Submodularity Characterization for the Mutually Complementary Case with Reconsideration]
        \label{thm:complementary_self_with_reconsideration}
            When the two ideas are complementary and reconsideration is allowed, for a fixed $S_B$, $\sigma_A$ is submodular in $S_A$ whenever one of the following holds:
            \begin{itemize}
                \item $\qa = \qab$,
                \item $\qb = \qba$,
                \item $\qb = 0$.
            \end{itemize}
            And when none of these conditions hold, submodularity is violated, i.e., there exists $(G, S_A, S_B, u, v)$ such that for each group of $(\qa, \qb, \qab, \qba)$ not satisfying the above conditions,
            \[
                \sigma_A(S_A, S_B) + \sigma_A(S_A \cup \{u, v\}, S_B) > \sigma_A(S_A \cup \{u\}, S_B) + \sigma_A(S_A \cup \{v\}, S_B).
            \]
        \end{theorem} 

        \begin{proof}
            We prove the negative part first. Consider the counterexample presented in Figure~\ref{fig:reconsideration_self}, and let the basic seed sets of $A$ and $B$ be $S_A = \{a_1, a_2\}$, $S_B = \{b\}$. We consider the marginals of $u$ as a $A$-seed when $v$ is a $A$-seed and when $v$ is not. Let
            \begin{align*}
                M_1 & = \sigma_A^t(S_A \cup \{u\}, S_B) - \sigma_A^t(S_A, S_B), \\
                M_2 & = \sigma_A^t(S_A \cup \{u, v\}, S_B) - \sigma_A^t(S_A \cup \{v\}, S_B).
            \end{align*}

            Note that the order of proposals at a vertex does not affect the final adoptation outcome \cite{lu2015competition}. We can therefore assign the orders of proposals in a way such that the marginals can be easily computed. In particular, $u$ has a non-zero marginal iff the following happen simultaneously in the order as listed:
            \begin{itemize}
                \item $a_1$ does not activate $u$, with probability $1 - \qa$;
                \item A $B$-proposal reaches $u$ and succeeds only when $u$ is already $A$-adopted, with probability $\qb \qba (\qba - \qb)$ when $v$ is not an $A$-seed and $\qba \qba (\qba - \qb)$ when $v$ is;
                \item Upon adoptation of $B$ by $u$, $a_2$ and $t$ will subsequently adopt $B$, with probability $\qba \qb$;
                \item $a_2$ tries to affect $t$ and succeeds only when $t$ is already $B$-adopted, with probability $\qab - \qa$.
            \end{itemize}
            We let the propagation corresponding to the conditions happen in exactly the order listed above. Formally, by multiplying the probabilities of the foregoing independent events, we have
            \begin{align*}
                M_1 & = (1 - \qa) \qb \qba (\qba - \qb) \qba \qb (\qab - \qa), \\
                M_2 & = (1 - \qa) \qba \qba (\qba - \qb) \qba \qb (\qab - \qa).
            \end{align*}
            Taking the difference,
            \[
                M_2 - M_1 = (1 - \qa) (\qba - \qb)^2 \qba^2 \qb (\qab - \qa).
            \]
            It is clear that when no conditions stated in Theorem~\ref{thm:complementary_self_with_reconsideration} hold, $M_2 - M_1 > 0$ and submodularity fails.\footnote{Note that when $A$ and $B$ are complementary, $\qba - \qb \ne 0 \Rightarrow \qba \ne 0$, and $\qab - \qa \ne 0 \Rightarrow \qa \ne 1$.}

            Now we look at the positive cases.
            \begin{itemize}
                \item
                    $\qa = \qab$. That means the propagation of $B$ does not help $A$ at all. Submodularity in this case trivially reduces to that in one-item IC model.
                \item
                    $\qb = \qba$. That means the propagation of $A$ does not affect $B$ at all. We can therefore let $B$ propagate first. When $B$ finishes propagating, the situation $A$ faces is just a generalized IC propagation procedure with possibly different vertex acceptance probabilities. Submodularity follows.
                \item
                    $\qb = 0$. That means $B$ does not propagate without the help of $A$, and therefore can never help $A$. Submodularity again reduces to that in IC model.
            \end{itemize}
        \qed
        \end{proof}

        \begin{figure}[t]
        \centering
        \begin{tikzpicture}
            \tikzset{vertex/.style = {shape=circle,draw,minimum size=2em}}
            \tikzset{edge/.style = {->,> = latex'}}
            \node[vertex] (b) at (0, 0) {$b$};
            \node[vertex] (v) at (2, 0) {$v$};
            \node[vertex] (a1) at (4, 0) {$a_1$};
            \node[vertex] (u) at (6, 0) {$u$};
            \node[vertex] (a2) at (8, 0) {$a_2$};
            \node[vertex] (t) at (10, 0) {$t$};
            \draw[edge] (b) to (v);
            \draw[edge] (v) to (a1);
            \draw[edge] (a1) to (u);
            \draw[edge] (u) to (a2);
            \draw[edge] (a2) to (t);
        \end{tikzpicture}
        \caption{Counterexample used in the proof of Theorem~\ref{thm:complementary_self_with_reconsideration}.}
        \label{fig:reconsideration_self}
        \end{figure}
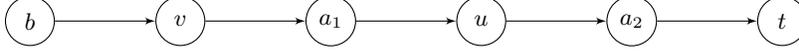

    \subsection{Cross Submodularity}
    
        Fixing $S_A$, because of the complementary nature of the two ideas, we are also curious about submodularity of $\sigma_A$ in $S_B$, i.e., submodularity of the influence of some idea w.r.t.\ the seed set of the other idea, fixing its own seed set.
        The following theorem provides the characterization in this case, for
        both with and without reconsideration.

        \begin{theorem}[Cross-Submodularity Characterization for the Mutually Complementary Case]
        \label{thm:complementary_cross}
            When the two ideas are complementary, no matter whether reconsideration is allowed or not, for a fixed $S_A$, $\sigma_A$ is submodular in $S_B$ whenever one of the following holds:
            \begin{itemize}
                \item $\qa = \qab$,
                \item $\qb = 1$.
            \end{itemize}
            And when none of these conditions hold, submodularity is violated, i.e., there exists $(G, S_A, S_B, u, v)$ such that for each group of $(\qa, \qb, \qab, \qba)$ not satisfying the above conditions,
            \[
                \sigma_A(S_A, S_B) + \sigma_A(S_A, S_B \cup \{u, v\}) > \sigma_A(S_A, S_B \cup \{u\}) + \sigma_A(S_A, S_B \cup \{v\}).
            \]
        \end{theorem}

        \begin{figure}[t]
        \centering
        \begin{tikzpicture}
            \tikzset{vertex/.style = {shape=circle,draw,minimum size=3em}}
            \tikzset{edge/.style = {->,> = latex'}}
            \node[vertex] (a) at (0, 0) {$a$};
            \node[vertex] (v) at (2, 0) {$v$};
            \node[vertex] (b) at (4, 0) {$b$};
            \node[vertex] (u/t) at (6, 0) {$u$ / $t$};
            \draw[edge] (a) to (v);
            \draw[edge] (v) to (b);
            \draw[edge] (b) to (u/t);
        \end{tikzpicture}
        \caption{Counterexample used in the proof of Theorem~\ref{thm:complementary_cross}.}
        \label{fig:cross}
        \end{figure}
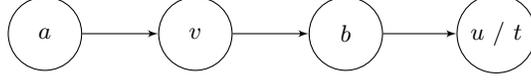

        \begin{proof}
            We prove the negative part first. Consider the counterexample presented in Figure~\ref{fig:cross} (where $u$ and $t$ are different names of the same vertex), and let the basic seed sets of $A$ and $B$ be $S_A = \{a\}$, $S_B = \{b\}$. We consider the marginals of $u$ as a $B$-seed when $v$ is a $B$-seed and when $v$ is not. Let
            \begin{align*}
                M_1 & = \sigma_A^t(S_A, S_B \cup \{u\}) - \sigma_A^t(S_A, S_B), \\
                M_2 & = \sigma_A^t(S_A, S_B \cup \{u, v\}) - \sigma_A^t(S_A, S_B \cup \{v\}).
            \end{align*}

            Node $u$ has a non-zero marginal iff an $A$-proposal reaches $t$ from $a$ and succeeds only when $t$ is $B$-adopted, while $t$ rejects the $B$-proposal from $b$. Since the order of proposals does not matter, w.l.o.g.\ we let $b$ make its proposal first and fail (with probability $1 - \qb$), and then $a$ propagate all the way to $t$ (with probability $\qa \qab (\qab - \qa)$ when $v$ is not a $B$-seed, and $\qab \qab (\qab - \qa)$ when $v$ is). Formally,
            \begin{align*}
                M_1 & = (1 - \qb) \qa \qab (\qab - \qa), \\
                M_2 & = (1 - \qb) \qab \qab (\qab - \qa).
            \end{align*}
            Taking the difference,
            \[
                M_2 - M_1 = (\qab - \qa)^2 (1 - \qb) \qab.
            \]
            It is clear that when none of the conditions stated in Theorem~\ref{thm:complementary_cross} hold, $M_2 - M_1 > 0$ and submodularity fails.\footnote{Note that when $A$ and $B$ are complementary, $\qab - \qa \ne 0 \Rightarrow \qab \ne 0$.}
            We further note that the above example works with or without
            	reconsideration.
            This is because that the reconsideration may only be triggered at
            	node $t$, and only when node $t$ initially does not adopt $B$, and
            	after it adopts $A$, it may reconsider adopting $B$.
            However, we are only looking at the adoption of $A$ at $t$ in $\sigma_A$,
            	and thus reconsideration of adopting $B$ at $t$ has no impact
            	in our analysis above.

            Now we look at the positive cases.
            \begin{itemize}
                \item
                    $\qa = \qab$. That means the propagation of $B$ does not help $A$ at all. Submodularity in this case trivially reduces to the
                    case of the one-item IC model.
                \item
                    $\qb = 1$. That means $B$ can affect any vertex it reaches, and $B$ propagation is indifferent to $A$'s adoption. 
                    We first discuss the case when reconsideration is allowed.
                    In this case, according to~\cite{lu2015competition}, whether
                    $A$ or $B$ arrives at a node first does not matter, and thus 
                    we can always assume that $B$ propagates first in the network,
                    and after $B$'s propagation ends, $A$ starts to propagate.
                    
                    We prove that for any possible world $W$, where $\sigma_A^t(S_A, S_B \cup \{u, v\}, W) = 1$, we have $\sigma_A^t(S_A, S_B \cup \{u\}, W) + \sigma_A^t(S_A, S_B \cup \{v\}, W) \ge 1$. That is, when $t$ is $A$-adopted when both $u$ and $v$ are $B$-seeds, $t$ will still be activated either when $u$ alone is a $B$-seed or $v$ alone is.

                    Let $p = \{w_1, \dots, w_k=t\}$ be the shortest path in the possible world $W$ through which $A$ affects $t$ when $u$ and $v$ are both $B$-seeds. 
                    Let $w$ be the closest vertex to $w_1$ on $p$ that adopts $B$.
                    If no such $w$ exists, then the argument is trivial, since
                    it means $A$ propagates to $t$ by itself, and thus we immediately
                    have $\sigma_A^t(S_A, S_B \cup \{u\}, W) = 
                    	\sigma_A^t(S_A, S_B \cup \{v\}, W) = 1$.
                    So we assume such $w$ exists.
                    Because $\qb=1$, all nodes after $w$ on path $p$ will also
                    adopt $B$, when $S_B\cup\{u,v\}$ is the $B$-seed set.
                    Let $p'$ be the path
                    in the possible world $W$ through which $B$ reaches $w$ from some $B$ seed. 
                    W.l.o.g.\ we assume that $v \notin p'$, and $p'$ starts 
                    	from some $B$-seed $x\in S_B\cup\{u\}$. 
                    We show that $\sigma_A^t(S_A, S_B \cup \{u\}, W) = 1$.
                    This is because in the possible world $W$, starting from $B$-seed
                    	$x\in S_B\cup \{u\}$, $x$ could reach $w$ and then $t$, and
                    	since $\qb=1$, all nodes along this path will adopt $B$.
                    Therefore, when $S_B\cup\{u\}$ is the $B$-seed set,
                    	it is the same that all nodes starting from $w$ on path $p$
                    	will adopt $B$, making it the same as the case
                    	when $S_B\cup\{u,v\}$ is the seed set.
                    Hence, $A$ propagates along the path $p$ in exactly the same way
                    	as if $S_B\cup\{u,v\}$ is the seed set, and thus $t$
                    	will adopt $A$ when $S_B\cup\{u\}$ is the $B$-seed set,
                    	namely, $\sigma_A^t(S_A, S_B \cup \{u\}, W) = 1$.
                    This is sufficient to show the cross-submodularity 
                    	of $\sigma_A$ with respect to $S_B$.
                    	
                    Now we discuss the case without reconsideration.
                    The argument follows the same structure as above.
                    The difference is now the order of item arrival at a node
                    	does matter, so we do not assume $B$ propagates first.
                    Instead, $A$ and $B$ propagate at the same time 
                    	according to the model.
                    On the path $p$, when we define $w$, now $w$ is the first
                    node from $w_1$ that adopts $B$ {\em before $A$ arrives}.
                    That means, for all nodes before $w$ in path $p$, 
                    	even if they adopt $B$, they adopt $B$ after adopting $A$,
                    	and since there is no reconsideration, these nodes
                    	adopt $A$ purely based on their $\qa$ condition, which further
                    	implies that these nodes will adopt $A$ in the possible
                    	world $W$ no matter what
                    	the $B$-seed set is.
                    Therefore, it also means that if no such $w$ exists, then
                    	we trivially have 
                    	$\sigma_A^t(S_A, S_B \cup \{u\}, W) = 
                    	\sigma_A^t(S_A, S_B \cup \{v\}, W) = 1$.
                    For all nodes following $w$ on path $p$, we claim that
                    	 $B$ arrives first before $A$ on these nodes, and thus
                    	 their adoption of $A$ is based on the condition $\qab$.
                    This is because $B$ arrives first at $w$ before $A$, so
                    	if $A$ propagates to the nodes after $w$ along the path $p$,
                    	then $A$ always arrives after $B$ at these nodes.
                    Thus if $A$ arrives first at some node $y$ after $w$, then going
                    	through $y$ there is a shorter path from $A$-seed set to $t$,
                    	contradicting the assumption that $p$ is the shortest
                    	path. 
                    Then, the rest argument follows the same discussion as above,
                    	showing that $w$ and all nodes after $w$ on path $p$
                    	will still adopt $B$ when $S_B\cup\{u\}$ is the $B$-seed
                    	set (w.l.o.g.), and thus $A$ could propagate along the
                    	path $p$ to reach $t$, just as in the case when
                    	$S_B\cup\{u,v\}$ is the $B$-seed set.
            \end{itemize}
        \qed
        \end{proof}

We remark that the result of Theorem~\ref{thm:complementary_cross} invalidates 
	Theorem~5 in~\cite{lu2015competition}, which claims that $\qba=1$ is a sufficient
	condition to guarantee cross-submodularity.
The proof of Theorem~5 in~\cite{lu2015competition} is incorrect, because 
	it does not consider the case that
	$B$ seeds may be on the path from an $A$ seed to a target node $v$, and by the Com-IC
	model a seed node assigned with $B$ will always adopt $B$, disregarding the $\qb$
	and $\qba$ parameters.
This is exactly what happens in the example given in Fig.~\ref{fig:cross}.
Thus, Claim 1 in the proof of Theorem~5 in~\cite{lu2015competition} is incorrect.
However, if the model would require that seed nodes also go through state transitions governed by
	the parameters $\qa, \qb, \qab, \qba$, just like other nodes during the propagation process,
	then Theorem~5 in~\cite{lu2015competition} would be correct.

\section{The One-Shot Model}

    In foregoing sections, properties of a model with somewhat rational agents are discussed. The agents are rational, in a sense that when a first proposal of some idea fails, they still allow the other idea (and sometimes even the first idea) a chance to propose; and when a first proposal succeeds, they do not accept/reject the possible proposal from the other idea instantly. In this section, we look at a model where agents act more extremely.

    \subsection{The Model}
    
        As in the Com-IC model, there is a backbone network $G = (V, E, p)$. 
        The model also has a number of parameters, depending on the number of ideas, as the GAP parameters in Com-IC.
        We only consider the mutually competing case for the One-Shot model.
        The key difference here is that an idle vertex considers only the first proposal that reaches it. When there are $m$ ideas $A_1, \dots, A_m$, each vertex has $m + 2$ possible states: idle, exhausted, $A_1$-adopted, \dots, $A_m$-adopted.
        
        Cascading proceeds in the following fashion: for any $i \in \{1, \dots, m\}$, when an $A_i$ proposal reaches an idle vertex, the vertex adopts $A_i$ w.p.\ $q_i$, and becomes exhausted w.p.\ $1 - q_i$. Once a vertex becomes exhausted, it no longer considers any further proposals. Since all ideas are competing against, an $A_i$-adopted vertex no longer considers proposals of $A_j$ where $j \ne i$. $(q_1, \dots, q_m)$ therefore completely characterizes the strengths of the ideas.

        \paragraph{Notations.} To accommodate numerous ideas, let $S_i$ be the seed set of $A_i$, $\sigma_i(S_1, \dots, S_m, W)$ be the number of vertices which adopt $A_i$ at the end of cascading in possible world $W$, and $\sigma_i(S_1, \dots, S_m)$ be the expectation of $\sigma_i(S_1, \dots, S_m, W)$ over possible worlds, etc.

    \subsection{Submodularity in One-Shot Model}
    
        The characterization of submodularity in One-Shot model appears to be more interesting. It demonstrates a dichotomy over the GAP space of One-Shot model, i.e., only the strongest idea propagates with submodularity.

        \begin{theorem}
        \label{thm:one-shot}
            In One-Shot model, for some $i \in \{1, \dots, m\}$, when $q_i \ge q_j$ for any $j \in \{1, \dots, m\}$ or $q_i = 0$, $\sigma_i$ is submodular in $S_i$; when there is some $j \in \{1, \dots, m\}$ such that $0 < q_i < q_j$, submodularity is violated. To be specific, when $0 < q_i < q_j$, there exists $(G, S_1, \dots, S_m, u, v)$ such that
            \begin{align*}
                & \sigma_i(S_1, \dots, S_m) + \sigma_1(S_1, \dots, S_i \cup \{u, v\}, \dots, S_m) \\
                >\ & \sigma_i(S_1, \dots, S_i \cup \{u\}, \dots, S_m) + \sigma_i(S_1, \dots, S_i \cup \{v\}, \dots, S_m).
            \end{align*}
        \end{theorem}

        \begin{figure}[t]
        \centering
        \begin{tikzpicture}
            \tikzset{vertex/.style = {shape=circle,draw,minimum size=3em}}
            \tikzset{edge/.style = {->,> = latex'}}
            \node[vertex] (u) at (0, 0) {$u$};
            \node[vertex] (j) at (4, 0) {$j$};
            \node[vertex] (v) at (4, -1.5) {$v$};
            \node[vertex] (x1) at (0, -1.5) {$x_1$};
            \node (x2) at (0, -3) {};
            \node (dx) at (0, -3.75) {\dots};
            \node (xk+1) at (0, -4.5) {};
            \node[vertex] (xk+2) at (0, -6) {$x_{k + 2}$};
            \node[vertex] (y1) at (4, -3) {$y_1$};
            \node (y2) at (4, -4.5) {};
            \node (dy) at (4, -5.25) {\dots};
            \node (yk-1) at (4, -6) {};
            \node[vertex] (yk) at (4, -7.5) {$y_k$};
            \node[vertex] (t) at (2, -9) {$t$};
            \draw[edge] (u) to (x1);
            \draw[edge] (x1) to (x2);
            \draw[edge] (xk+1) to (xk+2);
            \draw[edge] (xk+2) to (t);
            \draw[edge] (j) to (v);
            \draw[edge] (v) to (y1);
            \draw[edge] (y1) to (y2);
            \draw[edge] (yk-1) to (yk);
            \draw[edge] (yk) to (t);
        \end{tikzpicture}
        \caption{Counterexample used in the proof of Theorem~\ref{thm:one-shot}.}
        \label{fig:one-shot}
        \end{figure}
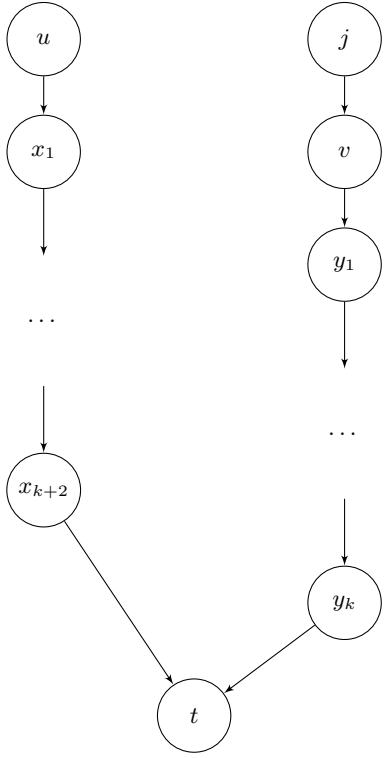

        \begin{proof}
            We prove the negative part first. Let $j$ be an item where $q_j > q_i$. Consider the network shown in Figure~\ref{fig:one-shot}, where the basic seed sets are $S_j = \{j\}$ and $S_k = \emptyset$ for any $k \ne j$. We calculate the marginals of $u$ at $t$ when $v$ is an $A_i$-seed and when $v$ is not. Formally, let
            \begin{align*}
                M_1 & = \sigma_i(S_1, \dots, S_i \cup \{u\}, \dots, S_m) - \sigma_i(S_1, \dots, S_m), \\
                M_2 & = \sigma_i(S_1, \dots, S_i \cup \{u, v\}, \dots, S_m) - \sigma_i(S_1, \dots, S_i \cup \{v\}, \dots, S_m).
            \end{align*}

            When $v$ is not a seed, $u$ has a positive marginal iff $j$ fails to reach $t$ and $u$ successfully reaches $t$. That is,
            \[
                M_1 = q_i^{k + 3} (1 - q_j^{k + 2}).
            \]
            And when $v$ is an $A_i$-seed, $t$ has a positive marginal iff $v$ fails to reach $t$ and $u$ succeeds. So,
            \[
                M_2 = q_i^{k + 3} (1 - q_i^{k + 1}).
            \]
            Taking the difference,
            \[
                M_2 - M_1 = q_i^{k + 3}(q_j^{k + 2} - q_i^{k + 1}).
            \]
            As $q_i < q_j$,
            \[
                \lim_{k \rightarrow \infty} \frac{q_j^{k + 2}}{q_i^{k + 1}} = \infty,
            \]
            so when $q_i > 0$, there is some $k$ such that $M_2 - M_1 > 0$, and submodularity is violated.

            We prove the positive part now. When $q_i = 0$, $\sigma_i = |S_i|$ is clearly submodular in $S_i$. Now we consider the other case. W.l.o.g.\ we renumber the items such that $q_1 \ge q_2 \ge \dots \ge q_m$, and show that $\sigma_1$ is submodular in $S_1$. We take a different possible world view here. Since each vertex considers only one proposal, it needs at most one random real number drawn uniformly at random from $[0, 1]$. When a $A_i$-proposal reaches a vertex $u$, $u$ accepts the proposal iff its random real number, denoted by $X_u$, does not exceed $q_i$, so effectively $u$ accepts a $A_i$-proposal w.p.\ $q_i$. Note that once $X_u$ is fixed, if $u$ accepts a $A_i$-proposal, it also accepts a $A_{i - 1}$-proposal given that it arrives first, because $X_u \le q_i \le q_{i - 1}$. Equivalently we may say that with probability $q_i - q_{i + 1}$ (where $q_0 = 1$ and $q_{m + 1} = 0$), $u$ accepts exactly the strongest $i$ proposals if they arrive first. We call these vertices type $i$ vertices. Each vertex belongs to exactly one of types $0$ through $m$.

            Consider a possible world interpretation where each possible world consists of the types of all vertices. We argue that in any possible world $W$, for any vertex $t$, $\sigma_i^t(S_1, \dots, S_m, W)$ is submodular in $S_1$, fixing $S_2$, \dots, $S_m$. To be specific, for any $S$, $u$, $v$, we show that if $t$ adopts $A_1$ when $S_1 = S \cup \{u, v\}$, then it must also adopt $A_1$ either when $S_1 = S \cup \{u\}$ or when $S_1 = S \cup \{v\}$. Remove all type $0$ vertices first, since they do not participate in the propagation. When $S_1 = S \cup \{u, v\}$, let $p = \{w_1, \dots, w_k\}$ be the shortest path through which $t$ is affected by $A_1$ , where $w_1 \in S \cup \{u, v\}$ and $w_k = t$. W.l.o.g.\ assume that $w_1 \ne v$. We show that $\sigma_1^t(S \cup \{u\}, \dots, S_m, W) = 1$. Assume the opposite, which implies that at least one vertex in $p$ is not $A_1$-adopted when $v$ is not a seed. Let $w$ be the vertex closest to $w_1$ on $p$, which becomes not $A_1$-adopted (and $A_i$-adopted instead) when $v$ is not a seed. $w$ must be reachable from $v$. Let $p' = \{x_1, \dots, x_l\}$ be the shortest path from $v$ to $w$, and $x$ the closest vertex to $v$ on $p'$ which becomes $A_i$-adopted when $v$ is not a seed. Since $v$ blocks $A_i$ from affecting $x$ through path $[x_1, x] \subseteq p'$, and when $v$ is not a seed, $x$ blocks $w$ from being affected by $A_1$ through path $[x, x_l] \subseteq p'$, clearly $p'$ is a shorter $A_1$-path (recall that $A_1$ can pass through every vertex in the world) from the $A_1$ seed set to $w$ than $[w_1, w] \subseteq p$ when $v$ is an $A_1$-seed, a contradiction.
        \qed
        \end{proof}
        
        \begin{note}
            Unlike all other theorems, the counterexample needed for Theorem~\ref{thm:one-shot} has to be constructed after fixing $q_i$ and $q_j$.
        \end{note}
        
\section*{Acknowledgment}

We would like to thank Yingru Li for some early discussions on the subject.
Wei Chen is partially supported by 
  the National Natural Science Foundation of China (Grant No. 61433014).
Hanrui Zhang is supported by NSF Award IIS-1527434.

\bibliographystyle{splncs03}
\bibliography{biblio}

\end{document}